\documentclass{article}

\title{\huge \bf{Limits on Relativistic Quantum Measurement} \\ \vspace{10mm} \Large Part III Mathematics Essay}

\author{\Large J. Fuksa}

\date{\normalsize\today}

\usepackage{amsmath}
\usepackage{amssymb}
\usepackage{amsfonts}
\usepackage{amsthm}
\usepackage{mathrsfs}
\usepackage{dsfont}
\usepackage{textgreek}
\usepackage{mathtools}

\usepackage{wrapfig}
\usepackage[title]{appendix}
\usepackage{xfrac}
\usepackage[load=physical]{siunitx}
\usepackage{braket}
\usepackage{bm}
\usepackage{enumitem}

\usepackage{graphicx}
\graphicspath{ {./images/} }

\usepackage[margin=0.8in]{geometry}
\usepackage{multicol}

\usepackage[sorting=none]{biblatex}
\DeclareMathAlphabet{\mathpzc}{OT1}{pzc}{m}{it}
\DeclareMathOperator{\Tr}{Tr}
\DeclareMathOperator\supp{supp}

\newtheorem{theorem}{Theorem}[section]

\newtheorem{lemma}[theorem]{Lemma}

\theoremstyle{definition}
\newtheorem{definition}{Definition}[section]

\theoremstyle{remark}
\newtheorem*{remark}{Remark}

\theoremstyle{claim}
\newtheorem{claim}{Claim}[section]

\usepackage[colorinlistoftodos]{todonotes}

\begin{document}
\maketitle
\pagenumbering{arabic}

\begin{abstract}
Requiring causality on measurements in quantum field theory seems to impose strong conditions on a self-adjoint operator to be really measurable. This may seem limiting and artificial in the operator language of algebraic quantum field theory (AQFT), but is essential for a truly relativistic theory. Recent publications attempt to deal with this issue by including the apparatus into the formalism, connecting AQFT with measurement theory, but other options have been suggested. In this essay, I discuss the causality conditions on self-adjoint operators both in the language of AQFT and in the language of quantum information theory. I then present measurement theory in AQFT, modelling the apparatus as a quantum field with coupling to the measured system restricted to a region of spacetime. I highlight how this approach leads to a causally well behaved theory. Finally, I attempt to formulate the causality conditions on measurements in the Feynman path integral approach, using the concept of decoherent histories. I claim that the path integral approach has problems with causality similar to the operator based approaches and that even here causality is an a posteriori condition. 
\end{abstract}

\tableofcontents

\pagebreak
\begin{multicols}{2}

\section{Introduction}
The fundamental differences between the language of quantum theories and general relativity have caused some serious issues when trying to come up with a fully relativistic quantum field theory, not to mention quantum gravity. The particular issue that I will be dealing with in this essay is that causality is in no way present in the standard quantum formalism, even in relativistic quantum field theories. This has been pointed out by Sorkin \cite{sorkin}. The issue lies in the way measurement is implemented using the projection postulate: it seemingly allows for superluminal signalling, even when non-selective measurements are considered. 

Recently, a measurement theory has been proposed by Fewster and Verch (the FV formalism) in the algebraic quantum field theory (AQFT) language, modelling the interaction between the measured system and the probe \cite{fewster}. In \cite{bostelmann} it has been shown that this in fact leads to a causally well behaved theory. Furthermore, it allows us to deal with the non-local and relativistically ambiguous character of the projection postulate. However, the FV formalism doesn't give us conditions on observables in QFT that would ensure causality.

After introducing some notation and terminology in section \ref{notation_term}, and then describing the Sorkin scenario, which is a particular setup of observers in spacetime which highlights the issues with causality (in section \ref{sorkin_scn}), I discuss the relationship between causality and Hilbert space quantum mechanics in section \ref{QM}. Motivated by the results of this section, I move on to quantum field theory in section \ref{causality_qft}. In section \ref{aqft} I introduce AQFT and relate its rather abstract language to the more familiar Hilbert space picture, using *-algebra representation theory. I then move on to discuss the FV framework in section \ref{FV-framework} and I show how it deals with causality. In section \ref{sum_over_hist}, I then discuss causality in the path integral formalism, using the concept of decoherent histories. This approach has been suggested in \cite{sorkin} and \cite{borsten}. I briefly summarise the conclusions in section \ref{conclusions}.

\section{Notation and terminology} \label{notation_term}
\subsection{Spacetime}
In this text, I will be studying causality in quantum systems in a globally hyperbolic Lorentzian manifold $M$ with time orientation and a metric $g$ with signature $+-...-$. Denote the collection of the manifold, metric and time orientation $\bm{M}$. The spacetime may be curved, but has to be a static background. I will now establish some notation and terminology I will be using throughout. 

The \textit{causal future/past} of $x \in M$ is the set of all points that can be reached from $x$ by causal future/past oriented curves from $x$ and is denoted $J^\pm (x)$. Note that $x \in J^\pm(x)$. For a region $S \subset M$ we write 
\begin{equation} J^\pm (S) \coloneqq \bigcup\limits_{x \in S} J^\pm(x) \: \text{ and } \: J(S) \coloneqq \J^+(S) \cup J^-(S). \end{equation}
The \textit{causal hull} of a region is $J^+(S) \cap J^-(S)$. If $S = J^+(S) \cap J^-(S)$, $S$ is called \textit{causally convex}. 

The \textit{causal complement} of a region $S$ is 
\begin{equation} S^\perp \coloneqq M \backslash J(S). \end{equation}
Regions $T$ and $S$ are \textit{causally disjoint} if $T \subset S^\perp$ (or equivalently $S \subset T^\perp$). If $J^+(S) \cap J^-(T)$ is empty, then there is a Cauchy surface to the future of $T$ and to the past of $S$. 

The \textit{future/past Cauchy development} of a subset $S \subset M$ is the set of points $p$, such that all past/future-inextendible causal curves through $p$ meet $S$. It is denoted $D^\pm(S)$. Define $D(S) \coloneqq D^+(S) \cup D^-(S)$. 

Define the ``in" region $M^-$ and the ``out" region $M^+$ of a subset $S \subset M$ by $M^\pm(S) \coloneqq M \backslash J^\mp(S)$.

\subsection{*-algebras and Hilbert spaces}
Given a Hilbert space $\mathcal{H}$, I will denote by $\mathcal{B}(\mathcal{H})$ the space of bounded linear operators on $\mathcal{H}$, which can be made a Hilbert space under the inner product
\begin{equation} \langle A, B \rangle = \Tr{(A^*B)}, \end{equation}
where $A^*$ denotes the adjoint. I will denote by $\mathcal{D}(\mathcal{H})$ the space of all density matrices, i.e. positive Hermitian operators with unit trace. 

One of the main mathematical objects in this text are *-algebras. A *-algebra $\mathpzc{A}$ is a complete normed algebraic vector space with an involution $^*: \mathpzc{A} \rightarrow \mathpzc{A}$, which preserves the algebraic operations. A special class are the C*-algebras, which satisfy the C* identity 
\begin{equation} \| A^*A \| = \| A \|^2. \end{equation}
For a detailed exposition, see e.g. \cite{algebras}.

\section{The Sorkin scenario} \label{sorkin_scn}
Traditionally, measurement is the centre point of quantum mechanics. Not only does it give us access to information about the abstract quantum state and hence allows us to interact with nature, it also itself changes the state of the system according to the projection postulate, so that sufficiently fast subsequent measurements of the same quantity give the same answers. Therefore the order of measurements is of fundamental importance. When trying to formulate quantum mechanics relativistically, this becomes an essential issue, since temporal order becomes observer dependent. As noticed by Sorkin in \cite{sorkin}, this problem manifests itself through worrying violations of causality. To analyse this, Sorkin proposed a setting, which I will refer to as the \textit{Sorkin scenario}.

Suppose three observers, Alice $\mathcal{O}_A$, Bob $\mathcal{O}_B$ and Charlie $\mathcal{O}_C$, located in regions $O_A$, $O_B$ and $O_C$ respectively, such that $O_A$ and $O_C$ are spacelike separated, but 
\begin{equation}
O_B \cap J^+(O_A) \neq \O \text{ and } O_B \cap J^-(O_C) \neq \O.
\end{equation}
This setup is illustrated in figure \ref{fig_sorkin}. We can now let Alice and Bob perform quantum operations $\mathcal{E}_A$ and $\mathcal{E}_B$ respectively and let Charlie measure an observable $A$. We can now study the conditions imposed on $\mathcal{E}_A$, $\mathcal{E}_B$ and $A$ by requiring that Charlie's measurement cannot depend on Alice's operation, i.e. that Alice cannot signal to Charlie. 

To see the main problem Sorkin scenario highlights, note that Bob's measurement fixes the causal ordering of the observers, even though Alice and Charlie are causally disjoint. That means that even if we demand that projectors associated with Alice's and Charlie's measurement commute, after inserting Bob's measurement in between, this will not be helpful. It is therefore quite easy to suggest a scheme in which Alice will be able to signal to Charlie, e.g., as noted by Sorkin in \cite{sorkin}, an incomplete Bell measurement.

\end{multicols}

\begin{figure}
\centering
\includegraphics{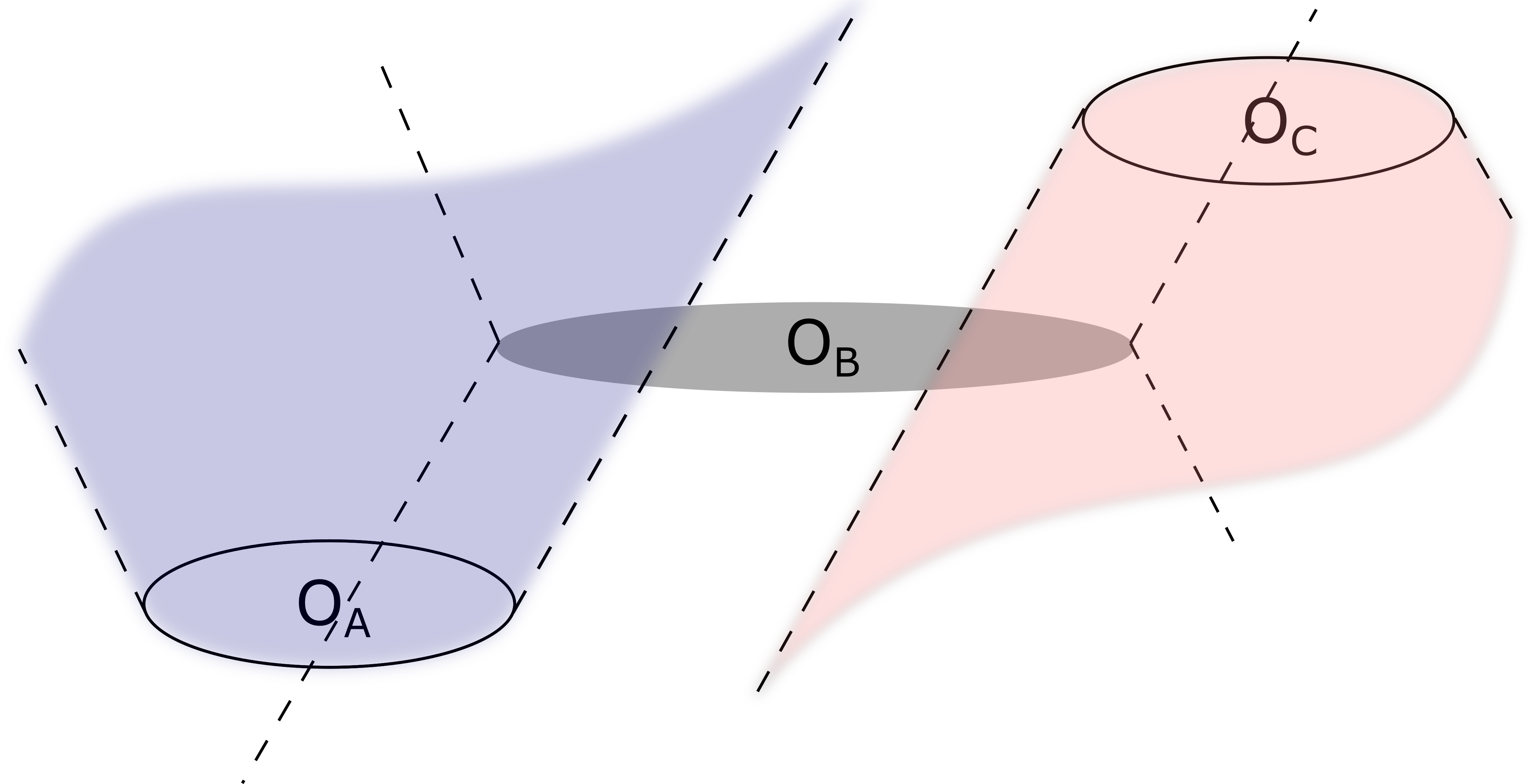}
\caption{\textbf{The Sorkin scenario.} $O_A, O_B$ and $O_C$ are regions controlled by Alice, Bob and Charlie respectively. The dashed lines are null curves. The blue region is $J^+(O_A)$, the red region is $J^-(O_C)$. Time is vertically upwards, one spatial dimension is on the horizontal axis.}
\label{fig_sorkin}
\end{figure}
\hrulefill
\begin{multicols}{2}

\section{Causality in quantum mechanics} \label{QM}
In this section I will discuss the Sorkin scenario in the context of Hilbert spaces and ordinary quantum mechanics.

\subsection{Dynamical variables vs. coordinates} \label{dyn_vars_vs_coords}
As pointed out by Hilgevoord in \cite{time}, accounts of quantum mechanics are notoriously sloppy in the distinction between position and time as dynamical variables $\bm{q}$, $\eta$ and as spacetime coordinates $\bm{x}$, $t$. Dynamical variables are coordinates in phase-space and are system-specific. For example a point-like particle will have a phase-space spanned by position $\bm{q}$ and momentum $\bm{p}$, whereas an ideal clock may be described by the dynamical time variable $\eta$ and its conjugate $\theta$. In quantum mechanics, each dynamical variable of a system is assigned a self-adjoint operator. There is no causal structure in phase-space, only trajectories parametrized by spacetime coordinates. On the other hand spacetime coordinates are system independent and don't have an associated operator in quantum mechanics. Crucially for our discussion, these are coordinates of the manifold $M$, which has a time orientation and a causal structure. 

Consider a particle in ordinary quantum mechanics. It is completely described by a wavefunction $\psi(\bm{q}, t).$\footnote{Note that here I use the notation $\bm{q}$ instead of the more usual $\bm{x}$ for the dynamical variable to avoid confusion with the coordinate, following Hilgevoord.} Here $\bm{q}$ is the dynamical variable, whereas $t$ is the spacetime coordinate. This is because the particle's state-space is the usual Hilbert space with a position and momentum operators; dynamics in time $t$ (which is a parameter and doesn't have an operator associated with it) is given by a Hamiltonian operator, according to the Schr\"{o}dinger equation. In this sense, quantum particle dynamics is a quantum field theory (QFT) in $0+1$ dimensions. Since $0+1$ dimensional Minkowski spacetime has trivial causal structure, ordering measurements poses no additional complication and we can unambiguously impose the projection postulate. This suggests that trying to make ordinary quantum mechanics relativistic in the Hilbert space picture is a lost cause.

\subsection{Causality from classical localization} \label{causality_from_classical_loc}
Since spacetime structure doesn't appear at all in Hilbert space quantum mechanics, we could resort to treating particle positions completely classically and quantize only their internal degrees of freedom. For example, we can consider qubits at spacetime locations $q_i = \mathrm{x}_i = (\bm{x}_i, t_i) \in M$ with quantum states $\ket{\psi_i} \in H_i$, where $H_i$ are the (for qubits 2-dimensional) Hilbert spaces of the qubit's internal state. We would now like to study the operations we can perform on the total Hilbert space of all qubits $\mathcal{H} = \bigotimes_i{H_i}$. 

In this framework however, since the positions of the qubits are treated classically, I have no other choice than to use the classical notion of locality. Therefore all operations on a system have to happen \textit{at a spacetime point}, in this approximation. Trying to study operations e.g. on Cauchy surfaces is going to lead to violations of causality almost by assumption, just as it does in the classical case: Consider a non-rotating light long \textit{rigid} rod in Minkowski spacetime with Alice and Bob located at each of its ends. Alice and Bob each have a large mass they can attach to the rod's end. If Bob pulls the rigid rod, by Newton's laws he can immediately find out, whether Alice attached her mass or not, so we get super-luminal signalling. The non-local rigidness of the rod, which here breaks causality, is the classical version of an operation on a Cauchy surface. When we treat position classically, it is \textit{locality} that secures causality, no matter whether we are performing quantum operations on the internal degrees of freedom of qubits or completely classical operations. \textit{Rigidity} is a well defined concept in relativistic physics, as pointed out by Max Born in \cite{rigidity}. The idea is that given a congruence $\mathcal{C}$ of timelike curves, which represents motion of the body, we get hypersurfaces orthogonal to the 4-velocities $u^a(x)$ of the worldlines in $\mathcal{C}$. In these hypersurfaces we can define a local space 3-metric $h_{ab}$. Rigidity is then defined by $(\mathcal{L}_u h)_{ab} = 0$ at each point of the body. This however requires a coordinated acceleration of each of the points and cannot be used to transmit information.

The question of which quantum operations are allowed on $\mathcal{H}$ when the qubits are separated in spacetime is therefore of practical importance: What spacetime setups of observers can I choose, so as to perform an operation $A$ on $\mathcal{H}$? What resources (e.g. shared entanglement) do they need? In no way does this question address the projection postulate or the notion of measurement in quantum mechanics. Locality and causality are completely classical in this approximation. After all, we could always imagine bringing the qubits together and putting all of them at the same time into a device performing the operation $A$. Since the worldlines of the qubits intersect in the device, using $A$ as the action of Bob in the Sorkin scenario ensures that Alice and Charlie will be timelike separated. Hence, $A$ is allowed to permit Alice to signal Charlie and the projection postulate can be applied unambiguously.

Suppose now that $n$ observers $\mathcal{O}_i$ at locations $\mathrm{y}_i = (\bm{y}_i, t')$,\footnote{By spacetime locations of an observer I mean the spacetime coordinates of the operation that they perform. I will use this slight abuse of terminology throughout this section.} each have access to a system with density matrix $\rho_i \in \mathcal{D}(H_{A_i})$. The system density matrices can be entangled, such that the total state in the Hilbert space $\mathcal{H}_A = \bigotimes_i{H_{A_i}}$ is $\rho$. Demand
\begin{equation}
\rho_i = \Tr_{A_{j\neq i}}{(\rho)}.
\end{equation}
Here the subscript $A_{j \neq i}$ on the trace denotes partial trace over all the systems, except for $A_i$. We wish to perform an operation $\mathcal{E}$ on $\mathcal{H}_A$. Following \cite{preskill} and generalizing slightly, I will now give conditions on the operation to be performable by a set of localized observers. My notion of performability connects the concepts of localizability and semilocalizability in \cite{preskill}.

First of all, define subsets of observers called \textit{causal chains}.

\begin{definition}[Causal chain]
Suppose a collection of observers $O = \{\mathcal{O}_i\}$ localized in a globally hyperbolic time oriented spacetime. A subset $C \subseteq O$ is called a \textit{causal chain}, if there exists a smooth timelike curve $\Gamma$ joining all observers in $C$.
\end{definition}

\begin{remark}
Each causal chain admits a unique \textit{causal order}, following the curve $\Gamma$ forwards in time.
\end{remark}

Let us cover the set of all observers $O = \{\mathcal{O}_i\}$ by causal chains $C_j \subseteq O$, such that every observer is in at least one chain, but one observer can be in multiple chains at once. This covering is not unique. I will provide observers in each causal chain with a joint ancilla state $\sigma_j \in \mathcal{D}(H_{B_j})$, which can be entangled into a state $\sigma \in \mathcal{D}(\mathcal{H}_B)$, where $\mathcal{H}_B = \bigotimes_j{H_{B_j}}$. Like for the system states, we demand
\begin{equation}
\sigma_i = \Tr_{B_{j \neq i}}{(\sigma)}.
\end{equation}
The observers in $C_k$ can send quantum information to each other in the causal order of $C_k$, which is why they have a common ancilla. Now each observer has access to their part of the system and to the ancillas of all causal chains they belong to. An observer $\mathcal{O}_k \in C_k$, however, can perform an operation on the ancilla $\sigma_k$ \textit{after} it has been acted on by all the observers in $C_k$ preceding $\mathcal{O}_k$ in the causal order on $C_k$.

Now let us causally order observers in the causal chains. Since if $\mathcal{O}_m$ precedes $\mathcal{O}_n$ in some causal chain, it will precede it in any other causal chain where they both happen to be,\footnote{This is a consequence of the assumption that the spacetime is globally hyperbolic.} we can use the causal orderings of $C_j$ to get a causal ordering on $O$ by demanding that if $\mathcal{O}_r$ precedes $\mathcal{O}_s$ in some causal chain, it also precedes it in the ordering on $O$. The ordering in $O$ is unique up to exchanges of observers who don't have access to any common ancillas.

This setup allows us to define a \textit{performable operation}.
\begin{definition}[Performability]
In the notation and setup from above, a quantum operation $\mathcal{E}$ on $\mathcal{H}_A$ is \textit{performable} by a set $O$ of $n$ observers, who are located in spacetime $\bm{M}$, if and only if there exists a covering of $O$ by causal chains $C_j$, such that in the causal ordering $O$ inherits from the causal orderings of $C_j$'s as above, $\mathcal{E}$ can be written as
\begin{equation}
\mathcal{E}(\rho) = \Tr_{B}{(\mathcal{E}_n \circ ... \circ \mathcal{E}_1)(\rho \otimes \sigma)},
\end{equation}
where each $\mathcal{E}_i$ is an operation that acts trivially on all the Hilbert spaces that cannot be accessed by $\mathcal{O}_i$.
\end{definition}

\begin{remark}
The non-uniqueness of ordering on $O$, given a covering $C_j$, does not pose any problems. This is because if the ordering of observers $\mathcal{O}_m$ and $\mathcal{O}_n$ is not determined by their order in some of the causal chains, they have access to different resources and hence the operations $\mathcal{E}_m$ and $\mathcal{E}_n$ will commute.
\end{remark}

Notice that this definition is highly dependent on the set $O$ and the positions of the observers. If for example we have just a single observer with access to all the resources, this definition does not restrict their actions at all.

\section{Causality in field theory} \label{causality_qft}

In order to get a fully quantum concept of causality, we have to somehow introduce the spacetime coordinates into the formalism. We need a QFT in more than $0+1$ dimensions. 

One of the ways forward is to get some motivation from quantum mechanics in the Heisenberg picture. We have a time independent state $\ket{\psi}$ in a Hilbert space $\mathcal{H}$ and a continuous homomorphism on the linear operators on $\mathcal{H}$ parametrized by time, the only spacetime coordinate we have in this theory. The homomorphism is defined by 
\begin{equation} \label{heisenberg}
\begin{split}
\frac{d}{dt} A(t) &= i [H, A(t)], \\
A(0) &= A,
\end{split}
\end{equation}
for all linear operators $A$, where $H$ is the Hamiltonian. If $A$ is some physical observable at $t = 0$, eq. \ref{heisenberg} gives us the corresponding operator at different times, i.e. in different regions of our $0+1$ dimensional spacetime. This can be seen as a motivation behind algebraic quantum field theory (AQFT).

The discussion from the previous section shows that causality doesn't come automatically in quantum theories, see e.g. \cite{sorkin} or \cite{borsten}. In particular, the Sorkin scenario highlights the central issue. We have to introduce a proper notion of locality first. It is not at all obvious how to do that in QFT. In the following, I will first define AQFT and then introduce the Fewster-Verch framework, which provides such concept of locality in AQFT, and show how this indeed yields a fully causal theory.

\subsection{AQFT} \label{aqft}
Since we would like to be able to add and multiply observables, they should form an algebraic structure. An AQFT associates a *-algebra $\mathpzc{A}(\bm{M})$ with a unit $\mathds{1}$, with the globally hyperbolic Lorentzian spacetime $\bm{M}$ with time orientation. $\mathpzc{A}$ inherits the topology of $\bm{M}$ in the following sense. Each causally convex open subset $N \subset M$ has an associated unital sub-*-algebra $\mathpzc{A} (\bm{M}; N) \subseteq \mathpzc{A}(\bm{M})$. The following conditions are imposed on the sub-*-algebras, in an attempt to ensure the correct causal and local behaviour of the theory.
\begin{enumerate}
	\item \textit{(Isotony)} If $N_1 \subset N_2$, then $\mathpzc{A}(\bm{M}; N_1) \subseteq \mathpzc{A}(\bm{M}; N_2)$. \par
	\item \textit{(Compatibility)} \label{compatibility} Providing $N \subset M$ with the metric and time-orientation inherited from $\bm{M}$, we get a spacetime $\bm{N}$. The AQFT associates with $\bm{N}$ a *-algebra $\mathpzc{A}(\bm{N})$. Compatibility demands that there is an injective unit-preserving algebraic *-homomorphism $\alpha_{\bm{M}; \bm{N}}: \mathpzc{A}(\bm{N}) \rightarrow \mathpzc{A}(\bm{M})$, with the subalgebra $\mathpzc{A}(\bm{M}; N)$ as its image. Furthermore, these maps have to obey $\alpha_{\bm{M_1}; \bm{M_2}} \circ \alpha_{\bm{M_2}; \bm{M_3}} = \alpha_{\bm{M_1}; \bm{M_3}}$, for all $M_3 \subset M_2 \subset M_1$. \par
	\item \textit{(Time-slice property)} Whenever $N$ contains a Cauchy surface, $\mathpzc{A}(\bm{M}; N) = \mathpzc{A}(\bm{M})$. \par
	\item \textit{(Einstein causality)} \label{einstein} If $N_1$ and $N_2$ are causally disjoint, then elements of $\mathpzc{A}(\bm{M}; N_1)$ commute with elements of $\mathpzc{A}(\bm{M}; N_2)$. \par
	\item \textit{(Haag property)} \label{haag} Let $K$ be a compact subset of $M$. If an element $A \in \mathpzc{A}(\bm{M})$ commutes with all elements of $\mathpzc{A}(\bm{M}; N)$ for all $N \subseteq K^\perp$, then $A \in \mathpzc{A}(\bm{M}; L)$, where $L$ is any connected open causally convex subset containing $K$.
\end{enumerate}

I will now give a brief motivation behind these conditions; for a thorough discussion see \cite{haag} or \cite{aqft}. The basic idea is that self-adjoint elements of $\mathpzc{A}(\bm{M}; N)$ represent observables that are localized in $N$ in the Heisenberg picture. With this in mind, condition 1) is natural. Condition 2) demands that submanifolds inherit their algebra from the parent algebra in a natural way. Condition 3) is imposed so that we recover the quantum mechanical picture of ``time" evolution between Cauchy surfaces, i.e. we have to be able to measure \textit{all} observables in a thin slice around each Cauchy surface. Finally the last two properties 4) and 5) are the attempt to ensure causality and locality in the theory. This will be discussed further in the following.

\subsubsection{States} \label{states}
We can now define states in AQFT.
\begin{definition}[States in AQFT]
Given a *-algebra $\mathpzc{A}(\bm{M})$, \textit{states} are linear functionals from $\mathpzc{A}(\bm{M})$ to $\mathbb{C}$, such that $\omega(A^*A) \geq 0 \: \forall A \in \mathpzc{A}(\bm{M})$. They are also normalized with respect to the unit, such that $\omega(\mathds{1}) = 1$. A state $\omega$ is \textit{pure} if and only if it cannot be written as $\omega = \lambda \sigma + (1-\lambda) \mu$ for some states $\sigma \neq \mu$ and $0 < \lambda < 1$.
\end{definition}

The outcome of the action of a state $\omega$ on a observable $A \in \mathpzc{A}$ is the expectation value of the measurement of $B$ on a quantum field in a state $\sigma$. All this is rather abstract, so I will now give a brief discussion of ways to connect it with more familiar QFT formalism in terms of \textit{representations}. A more thorough summary can be found in \cite{aqft}.

\begin{definition}[*-algebra representations] 
A \textit{representation} of a unital *-algebra $\mathpzc{A}$ is a triple $(\mathcal{H}, \mathcal{D}, \pi)$, where $\mathcal{H}$ is a Hilbert space, $\mathcal{D}$ a dense subspace of $\mathcal{H}$ and $\pi$ a map from $\mathpzc{A}$ to linear operators on $\mathcal{H}$, such that
\begin{enumerate}
	\item $\mathcal{D}$ is the domain of $\pi(A) \: \forall A \in \mathpzc{A}$ and the range of each $\pi(A)$ is contained in $\mathcal{D}$,
	\item $\pi(\mathds{1}) = \mathds{1}$, the identity on $\mathcal{H}$,
	\item $\pi(A + \lambda B + CD) = \pi(A) + \lambda \pi(B) + \pi(C) \pi(D)$, so that the representation respects the algebra operations,
	\item each $\pi(A)$ has an adjoint $\pi(A)^*$, which has a domain containing $\mathcal{D}$, and
	\item $\pi(A)^*$ restricted to $\mathcal{D}$ is equal to $\pi(A^*)$.
\end{enumerate}
\end{definition}
A representation is \textit{faithful} if $\ker{\pi} = \O$; and is \textit{irreducible} if there are no subspaces of $\mathcal{H}$ invariant under $\pi(A)$, which are neither trivial, nor dense. 

Suppose $\pi$ has image that is dense in the space of linear operators on $\mathcal{H}$. Now, given a representation $(\mathcal{H}, \mathcal{D}, \pi)$ of $\mathpzc{A}$, the closure of
\begin{equation} \tilde{\rho}(\omega)(\pi(A)) \coloneqq \omega(A) \end{equation}
defines for each state $\omega$ on $\mathpzc{A}$ a linear functional $\tilde{\rho}(\omega)$ on the space of linear operators on $\mathcal{H}$. To see that $\tilde{\rho}(\omega)$ is indeed a linear functional, consider
\begin{equation}
\begin{split}
\tilde{\rho}(\omega)\big(\pi(A) + \lambda \pi(B)\big) &= \tilde{\rho}(\omega)\big(\pi(A + \lambda B)\big) \\
&= \omega(A + \lambda B) \\
&= \omega(A) + \lambda \omega(B) \\
&= \tilde{\rho}(\omega)(\pi(A)) + \lambda \tilde{\rho}(\omega)(\pi(B)) 
\end{split}
\end{equation}
using linearity of the states and the *-algebra representation properties. Since the image of $\pi$ is assumed dense, this proves the linearity of $\tilde{\rho}(\omega)$. This now allows us to use the Riesz representation theorem on the Hilbert space of bounded linear operators $\mathcal{B}(\mathcal{H})$ with inner product
\begin{equation} \langle A, B \rangle \coloneqq \Tr{(A^*B)}, \end{equation}
to get an operator $\rho$ corresponding to each $\tilde{\rho}$. This operator is positive by the positivity of $\omega$ and the properties of *-algebra representation. It is also Hermitian, since $\forall A \in \mathpzc{A}$ with $\pi(A) \in \mathcal{B}(\mathcal{H})$
\begin{equation} \begin{split}
\Tr{(\rho \pi(A)^* \pi(A))} &= \Tr{(\rho \pi(A^*) \pi(A))} \\
&= \omega (A^* A) \\
&= \big[\omega(A^* A)\big]^* \\
&= \big[\Tr{(\rho \pi(A)^* \pi(A))}\big]^* \\
&= \Tr{(\rho^* \pi(A)^* \pi(A))},
\end{split} \end{equation}
which implies $\rho = \rho^*$. It also has unit trace, as we can see form
\begin{equation} \Tr{(\rho)} = \langle \rho, \mathds{1} \rangle = \tilde{\rho}(\omega)(\mathds{1})= \omega (\mathds{1}) = 1. \end{equation}
Hence $\rho$ is a density operator, which connects the AQFT formalism with the standard Hilbert space language.

It is however important to note that there is a lot of freedom in choosing the representation, so it may still be difficult to interpret the results in AQFT. One possible natural choice of representation is provided by the Gelfand, Naimark, Segal (GNS) theorem.
\begin{theorem}[GNS theorem] \label{gns}
Let $\omega$ be a state on a *-algebra $\mathpzc{A}$. Then there is a representation $(\mathcal{H}_\omega, \mathcal{D}_\omega, \pi_\omega)$ of $\mathpzc{A}$ and a cyclic unit vector $\Omega_\omega \in \mathcal{D}_\omega$, such that
\begin{equation} \omega(A) = \braket{\Omega_\omega | \pi_\omega(A) \Omega_\omega} \: \forall A \in \mathpzc{A}.
\end{equation}
The quadruple $(\mathcal{H}_\omega, \mathcal{D}_\omega, \pi_\omega, \Omega_\omega)$ is unique up to unitary equivalence.

If $\mathpzc{A}$ is in fact a C*-algebra, then furthermore
\begin{enumerate}
	\item each $\pi_\omega(A)$ extends to a bounded operator on $\mathcal{H}_\omega$,
	\item $\omega$ is pure if and only if the representation is irreducible and
	\item if the representation is faithful, then $\|\pi_\omega(A)\| = \|A\|_\mathpzc{A}$.
\end{enumerate}
\end{theorem}
Here cyclic means that $\mathcal{D}_\omega = \pi_\omega(\mathpzc{A}) \Omega_\omega$. For proof of Theorem \ref{gns}, see \cite{aqft}. From property 2) for C*-algebras in Theorem \ref{gns} we see that $\Omega_\omega$ is a pure state on a larger Hilbert space than the natural one for the field theory.

\subsubsection{Operations} \label{operations}
In this section I would like to show how operations on quantum fields are performed in the AQFT language. Following \cite{fewster}, I allow observers to interact with the quantum field, which I will call the system, via another quantum field, the probe, which will be coupled to the system in some compact spacetime region $K \subset M$.

To make this more concrete, we may consider the system to be a quantum field with action 
\begin{equation} S_S[\Psi] = \int d^4x \mathcal{L}_S(\Psi) \end{equation} 
and the probe a quantum field with action 
\begin{equation} S_P[\Phi] = \int d^4x \mathcal{L}_P(\Phi). \end{equation} 
Both of these Lagrangians may include self-interactions and overall are not restricted in any way. I am in principle not even limited to fields with actions that can be written in terms of a Lagrangian density; but I choose to do so for concreteness. Now I introduce an interaction between the fields of the form 
\begin{equation} \label{interaction} S_{int}[\Psi, \Phi] = \int d^4x \rho\Psi\Phi, \end{equation}
where $\supp{\rho} \subset K$. This is, once again, just a choice.

Now I would like to find the operation that is induced by this interaction. First, I will follow the analysis by Fewster and Verch \cite{fewster} to get the scattering morphism.

Let $\mathpzc{A}(\bm{M}) \otimes \mathpzc{B}(\bm{M})$ be the *-algebra of the uncoupled system-probe theory and $\mathpzc{C}(\bm{M})$ be the *-algebra of the coupled theory. Let the AQFT inclusion maps (as in condition \ref{compatibility} in section \ref{aqft}) for region $N$ in these theories be $\alpha_{\bm{M}; N} \otimes \beta_{\bm{M}; N}$ and $\gamma_{\bm{M}; N}$ respectively. Since the coupling is localized in $K$, the theory $\mathpzc{C}(\bm{M})$ will reduce to $\mathpzc{A}(\bm{M}) \otimes \mathpzc{B}(\bm{M})$ outside the causal hull $J^+(K) \cap J^-(K)$ of $K$. Hence, in any causally convex region $L \subset M \backslash J^+(K) \cap J^-(K)$ there will be an isomorphism 
\begin{equation} \chi_{\bm{L}}: \mathpzc{A}(\bm{L}) \otimes \mathpzc{B}(\bm{L}) \rightarrow \mathpzc{C}(\bm{L}). \end{equation}
In particular, an important role will be played by this morphism for $L = M^\pm(K)$, where as before, $M^\pm(K) = M \backslash J^\mp(K)$ are the ``in" (-) and ``out" (+) regions of $K$. Define $\chi^\pm \coloneqq \chi_{M^\pm(K)}$ for convenience. Define further in a similar fashion the inclusion maps for $M^\pm$ to be $\alpha^\pm, \beta^\pm, \gamma^\pm$ in the obvious way. Since $K$ is compact, regions $M^\pm$ contain a Cauchy surface\footnote{For a proof of this statement see e.g. \cite[Lemma A.4]{fewster_2012}. Geometrically however, this claim is quite well motivated, since e.g. the causal future and present of a compact region should contain a part of any timelike curve.} and hence $\alpha^\pm, \beta^\pm, \gamma^\pm, \chi^\pm$ are isomorphisms by the AQFT axioms.

This allows me to define, following \cite{fewster}, isomorphisms
\begin{equation} \kappa^\pm \coloneqq \gamma^\pm \circ \chi^\pm: \mathpzc{A}(\bm{M}^\pm) \otimes \mathpzc{B}(\bm{M}^\pm) \rightarrow \mathpzc{C}(\bm{M}) \end{equation}
and the retarded (+) and advanced (-) response maps
\begin{equation} \tau^\pm = \kappa^\pm \circ (\alpha^\pm \otimes \beta^\pm)^{-1}: \mathpzc{A}(\bm{M}) \otimes \mathpzc{B}(\bm{M}) \rightarrow \mathpzc{C}(\bm{M}). \end{equation}
The response maps relate the uncoupled theory to the coupled one through the identification at early (-) or late (+) times. This now gives the scattering morphism,
\begin{equation} \Theta \coloneqq (\tau^-)^{-1}\tau^+ \end{equation}
which is an automorphism of $\mathpzc{A}(\bm{M}) \otimes \mathpzc{B}(\bm{M})$ and relates the uncoupled theory identified with the coupled one at late times with the uncoupled theory identified with the coupled one at early times. The action of this morphism is similar to that of a scattering matrix in standard formulation of QFT.

I will also define the adjoint action of this map on the states through
\begin{equation} \omega(\Theta(O)) \eqqcolon \Theta^*(\omega)(O) \;\;\; \forall O \in \mathpzc{A}(\bm{M}) \otimes \mathpzc{B}(\bm{M}), \end{equation}
where $\omega$ is any state on the uncoupled algebra. The map $\Theta^*$ relates states at early times to states at late times. Now I will use the representation of states to relate this map to a quantum operation on a density matrix in a Hilbert space.

Let us pick a representation $(\mathcal{H}, \mathcal{D}, \pi)$ for $\mathpzc{A}(\bm{M}) \otimes \mathpzc{B}(\bm{M})$. In the way described in section \ref{states}, this defines a density matrix for each state on $\mathpzc{A}(\bm{M}) \otimes \mathpzc{B}(\bm{M})$. Hence, given a density matrix $\rho$ for an initial state $\omega$ on the decoupled system, the final state $\Theta^*(\omega)$ gives a new density matrix $\rho'$. This defines an operation $\mathcal{E}: \mathcal{D}(\mathcal{H}) \rightarrow \mathcal{D}(\mathcal{H})$, such that $\rho \mapsto \rho'$. Hence, picking a representation, we get the quantum operation corresponding to the interaction of the system and probe quantum fields.

In the following lemma, I give three crucial properties of the scattering morphism, which are proven in appendix A of \cite{fewster}.

\begin{lemma}[Proposition 3.1 in \cite{fewster}] \label{scattering_properties}
A scattering morphism $\Theta$ on the system-probe theory $\mathpzc{U}(\bm{M}) \coloneqq \mathpzc{A}(\bm{M}) \otimes \mathpzc{B}(\bm{M})$, which arises due to the interaction of the system with the probe in a compact region $K \subset M$, has the following properties:
\begin{enumerate}
	\item If $\hat{\Theta}$ is the scattering morphism which is obtained if we replace in the derivation the interaction region $K$ with a compact region $\hat{K} \supset K$, $\hat{\Theta} = \Theta$.
	\item If $L \subset K^\perp$, then $\Theta$ acts trivially on $\mathpzc{U}(\bm{M}; L)$.
	\item Suppose that $L^+ \subset M^+(K)$ and $L^- \subset M^-(K)$ are open, causally convex subsets, such that $L^+ \subset D(L^-)$. Then $\Theta \mathpzc{U}(\bm{M}; L^+) \subseteq \mathpzc{U}(\bm{M}; L^-)$.
\end{enumerate}
\end{lemma}

These properties follow from the AQFT axioms from section \ref{aqft}. I will omit the proof here, but give some insight into these properties.

The first property shows that there is an ambiguity in the definition of the interaction region K. For example, returning to the example of an interaction in eq. \ref{interaction}, we can choose $K$ to be any region containing $\supp{\rho}$, without changing the dynamics.

Causality makes an appearance through the second property. Causality demands, that the interaction localized in $K$ cannot influence the behaviour in $K^\perp$. This means that when we map the observables in $K^\perp$ from the uncoupled theory identified with the coupled one at late times to the uncoupled theory identified with the coupled one at early times, we should get the same observable.

The third property says that physics in $L^+$ is completely determined by the physics in $L^-$, as it should be, since $L^+$ is a part of the Cauchy development of $L^-$.

\subsection{The FV framework} \label{FV-framework}
In this section I will summarize the Fewster-Verch (FV) framework, presented in \cite{fewster}. This framework introduces measurement theory into AQFT by proposing a scheme that allows an observer with access to a probe quantum field to measure an observable on the system quantum field. The projection postulate is assumed on the probe measurement, so this approach does not attempt to solve the measurement problem. However, its importance lies in how elegantly it deals with the causality and locality problem in QFT.

\subsubsection{The measurement scheme} \label{FV-measurement_scheme}
In the notation from section \ref{aqft}, we have two quantum fields: system and probe, which are coupled only in a compact region $K$. We have the uncoupled *-algebra $\mathpzc{A}(\bm{M}) \otimes \mathpzc{B}(\bm{M})$ and the coupled *-algebra $\mathpzc{C}(\bm{M})$. We have the maps $\tau^\pm, \kappa^\pm$ and the scattering morphism $\Theta$ as before. Suppose that our system and probe states are uncorrelated at early times. That means that
\begin{equation} (\tau^-)^*\tilde{\omega} = \omega \otimes \sigma, \end{equation}
where $\tilde{\omega}$ is a state on $\mathpzc{C}(\bm{M})$, $\omega$ a state on $\mathpzc{A}(\bm{M})$ and $\sigma$ a state on $\mathpzc{B}(\bm{M})$. Suppose we measure the observable $B \in \mathpzc{B}(\bm{M})$ on the probe at late times. This corresponds to the observable
\begin{equation} \tilde{B} = \tau^+ (\mathds{1} \otimes B) \end{equation}
on $\mathpzc{C}(\bm{M})$. The expectation value of such measurement is hence
\begin{equation} 
\begin{split}
\tilde{\omega}(\tilde{B}) &= \big[(\tau^-)^{-1}\big]^* (\omega \otimes \sigma) \big(\tau^+ (\mathds{1} \otimes B)\big) \\
&= (\omega \otimes \sigma) \big((\tau^-)^{-1} \tau^+(\mathds{1} \otimes B)\big) \\
&= (\omega \otimes \sigma) \big(\Theta(\mathds{1} \otimes B)\big).
\end{split}
\end{equation}

Now I can define the concept of an induced observable. This is the observable $A \in \mathpzc{A}(\bm{M})$ that we would like to get information about through the measurement of $B$. Hence, we require that our measurement scheme effectively just evaluates $A$ on the system state. More precisely, demand
\begin{equation} \label{induced_obs} \tilde{\omega}(\tilde{B}) = \omega(A). \end{equation}
Fewster and Verch find a unique solution of eq. \ref{induced_obs} by defining two maps. First, define $\eta_\sigma: \mathpzc{A}(\bm{M})\otimes\mathpzc{B}(\bm{M}) \rightarrow \mathpzc{A}(\bm{M})$ by 
\begin{equation} \eta_\sigma(A \otimes B) = \sigma(B) A \end{equation}
and extending by linearity. 

Further, define the map $\varepsilon_\sigma: \mathpzc{B}(\bm{M}) \rightarrow \mathpzc{A}(\bm{M})$ by
\begin{equation} \varepsilon_\sigma (B) = (\eta_\sigma \circ \Theta)(\mathds{1} \otimes B). \end{equation}

Now we can check
\begin{equation}
\begin{split}
\omega\big(\varepsilon_\sigma(B)\big) &= \omega\big((\eta_\sigma \circ \Theta)(\mathds{1} \otimes B)\big) \\
&= (\omega \otimes \sigma)\big(\Theta(\mathds{1} \otimes B)\big) \\
&= \tilde{\omega}(\tilde{B}),
\end{split}
\end{equation}
so $\varepsilon_\sigma(B) \in \mathpzc{A}(\bm{M})$ is the induced system observable of the probe observable $B$. This construction therefore provides us with a measurement scheme. 

Notice that the induced observable depends not only on the probe observable $B$, but also on the probe initial state $\sigma$.

\paragraph{Localization of the induced observable.} Fewster and Verch in \cite{fewster} show that the induced observable can be localized in any connected open causally convex set containing $K$. 
\begin{theorem}[Localization of induced observables, Theorem 3.3 in \cite{fewster}]
For any probe observable $B \in \mathpzc{B}(\bm{M})$, the induced observable $\varepsilon_\sigma(B)$ can be localized in any connected open causally convex set containing the interaction region $K$.
\end{theorem}
\begin{proof}
Suppose a region $L \subseteq K^\perp$, $A \in \mathpzc{A}(\bm{M}; L)$ and $B \in \mathpzc{B}(\bm{M})$. Now
\begin{equation}\begin{split}
[\varepsilon_\sigma(B), A] &= [\eta_\sigma(\Theta(\mathds{1} \otimes B)), A] \\
&= \eta_\sigma[\Theta(\mathds{1} \otimes B), A \otimes \mathds{1}] \\
&= \eta_\sigma(\Theta[\mathds{1} \otimes B, A \otimes \mathds{1}]) \\
& = 0,\\
\end{split}\end{equation}
where I used point 2. in Lemma \ref{scattering_properties} to get the third equality. Therefore by the Haag property \ref{haag} in section \ref{aqft}, $\varepsilon_\sigma(B)$ can be localized in any connected open causally convex region containing $K$.
\end{proof}

\paragraph{Effect valued measure} The concept of a positive operator valued measure (POVM) is generalized to the *-algebraic setting by the notion of an effect valued measure (EVM). 
\begin{definition}[Effect valued measure] \label{evm}
Suppose $\mathpzc{A}(\bm{M})$ is a *-algebra and $\chi$ is a \textsigma-algebra of subsets of a set $\Omega$. An \textit{effect valued measure} (EVM) is a map $E: \chi \rightarrow \mathpzc{A}(\bm{M})$, which satisfies
\begin{enumerate}
	\item $E(X) \geq 0 \; \forall X \in \chi$,
	\item $E(\Omega) = \mathds{1}$ and
	\item given a set $\{X_i\}_i$, such that $X_i \in \chi \; \forall i$ and $X_i \cap X_j = \O \; \forall i \neq j$, 
	\begin{equation} E(\bigcup\limits_i X_i) = \sum\limits_i E(X_i). \end{equation}
\end{enumerate}
Call an EVM a \textit{projective effect valued measure} (PEVM), if furthermore
\begin{enumerate}[resume]
	\item $E(X)E(X) = E(X) \; \forall X \in \chi$ and
	\item $E(X)E(Y) = E(X \cap Y) \; \forall X, Y \in \chi$.
\end{enumerate}
\end{definition}
Physical reason for using EVMs is to allow for simultaneous measurement of some non-commuting observables (see \cite{dorofeev} and \cite{simultaneous}). For a thorough exposition see \cite{measurement}. The set $\Omega$ represents possible values of some observable $A \in \mathpzc{A}(\bm{M})$; given $X \in \chi$, the operator $E(X)$ represents the experimental result that the value of $A$ lies in the subset $X \subseteq \Omega$, in the sense that for a system in state $\omega$, the value of $A$ will be in $X$ with probability $\omega(E(X))$. This is analogous to how POVMs are used in quantum theory in general. 

Note that given an EVM on the probe, the induced observable on the system will also be an EVM, since $\varepsilon$ is linear and positivity preserving.

\paragraph{The projection postulate} Suppose observables $A \in \mathpzc{A}(\bm{M})$ and $B \in \mathpzc{B}(\bm{M})$ with associated EVMs $E_A: \chi_A \rightarrow \mathpzc{A}(\bm{M})$ and $E_B: \chi_B \rightarrow \mathpzc{B}(\bm{M})$ respectively. We would like to know what is the probability that the value of $A$ will be measured to be in $X \in \chi_A$, given that the value of $B$ has been measured to be in $Y \in \chi_B$, if the system and probe are initially in a state $\omega \otimes \sigma$. By the definition of conditional probability,\footnote{Here I am abusing the notation $A \in X$ to denote that the value of the observable $A$ is in $X$.}
\begin{equation} \begin{split} \label{projection_post}
P(A \in X | &B \in Y) = \frac{P(A \in X \& B \in Y)}{P(B \in Y)} \\
&= \frac{(\omega \otimes \sigma)(\Theta(E_A(X) \otimes E_B(Y)))}{\sigma(E_B(Y))} \\
&= \frac{J_\sigma(E_B(Y))(\omega)(E_A(X))}{J_\sigma(E_B(Y))(\omega)(\mathds{1})},
\end{split}\end{equation}
where $J_\sigma(B)(\omega)$ for $B \in \mathpzc{B}(\bm{M})$, $\omega$ a state on $\mathpzc{A}(\bm{M})$ and $\sigma$ a state on $\mathpzc{B}(\bm{M})$ is defined by its action on any $A \in \mathpzc{A}(\bm{M})$
\begin{equation} J_\sigma(B)(\omega)(A) \coloneqq (\omega \otimes \sigma)(\Theta(A \otimes B)). \end{equation}
We call the map $J_\sigma(B)$ the \textit{pre-instrument}.

We can interpret eq. \ref{projection_post} in the language of the projection postulate. The updated state after the measurement of $B \in Y$ becomes
\begin{equation} \tilde{\omega} \coloneqq \frac{J_\sigma(E_B(Y))(\omega)}{J_\sigma(E_B(Y))(\omega)(\mathds{1})}, \end{equation}
such that in a subsequent measurement of $A$, the probability of obtaining a value in $X \in \chi_A$ is given by
\begin{equation} P(A \in X) = \tilde{\omega}(E_A(X)). \end{equation}
It is clear that $J_\sigma(E_B(Y))(\omega)$ is the unnormalized updated system state.

\subsubsection{Causality in the FV framework} \label{FV-causality}
The main point of the FV framework is that the system-probe interaction is localized in $K$. Hence, recovering a quantum operation from this interaction in the sense of section \ref{operations} provides us with a notion of locality in AQFT. 

As was the case for quantum mechanics in section \ref{QM}, \textit{causality} emerges from \textit{locality} in quantum field theory too, as has been shown in \cite{bostelmann} and in \cite{fewster}. In this section, I summarize these results. 

\paragraph{Causality and post-selection} Consider an observer $\mathcal{O}_1$, who measures a probe observable associated with an EVM $E_B: \chi_B \rightarrow \mathpzc{B}(\bm{M})$, and an observer $\mathcal{O}_2$, who measures a system observable associated with an EVM $E_A: \chi_A \rightarrow \mathpzc{A}(\bm{M})$. If the observers are not allowed any other communication, observer $\mathcal{O}_2$ cannot know the result of the measurement $\mathcal{O}_1$ performs. Hence $\mathcal{O}_2$ can be considered to perform their measurement on the sum of all the unnormalized updated states corresponding to a mutually exclusive complete set of possible results that $\mathcal{O}_1$ can obtain. This can be written, using linearity of the pre-instruments, as
\begin{multline} \tilde{\omega} = \sum_{Y_i \in \kappa} J_\sigma(E_B(Y_i))(\omega) = \\
= J_\sigma(E_B(\Omega_B))(\omega) = J_\sigma(\mathds{1})(\omega),\end{multline}
where $\Omega_B$ is the set corresponding to $\chi_B$ and $\kappa$ is a subset of $\chi_B$ such that for all $\alpha \neq \beta \in \kappa$ we have that $\alpha \cap \beta = \O$ and $\bigcup_{\alpha \in \kappa} \alpha = \Omega_B$. This state is normalized, since
\begin{equation}\begin{split} \tilde{\omega}(\mathds{1}) &= J_\sigma(\mathds{1})(\omega)(\mathds{1}) \\
&= (\omega \otimes \sigma)(\Theta(\mathds{1} \otimes \mathds{1})) \\
&= (\omega \otimes \sigma)(\mathds{1} \otimes \mathds{1}) \\
&= 1. 
\end{split}\end{equation}

Consider what happens if in fact $E_A(X) \in \mathpzc{A}(\bm{M}; L) \; \forall X \in \chi_A$, where $L \subset K^\perp$. Now
\begin{equation}\begin{split}
\tilde{\omega}(E_A(X)) &= J_\sigma(\mathds{1})(\omega)(E_A(X)) \\
&= (\omega \otimes \sigma)(\Theta(E_A(X) \otimes \mathds{1})) \\
&= (\omega \otimes \sigma)(E_A(X) \otimes \mathds{1}) \\
&= \omega(E_A(X)),
\end{split}\end{equation}
so the probability distribution of the measurement $\mathcal{O}_2$ performs is unchanged by the fact that $\mathcal{O}_1$ performed their measurement. This is in agreement with causality, since now the observers are measuring observables that can be localized in causally disjoint regions.

\paragraph{Sorkin scenario in the FV framework.} First, it is necessary to formulate how observers in different causal relationships will be represented in the AQFT and FV language. In particular, we are interested in the Sorkin scenario.

Each observer $\mathcal{O}_i$, $i \in \{A, B, C\}$ has access to a probe AQFT with an associated *-algebra $\mathpzc{B}_i(\bm{M})$. Each probe interacts with the system AQFT in a compact spacetime region $K_i \subseteq O_i$. Similarly to section \ref{operations}, this gives us a coupled and an uncoupled theory. I will use a shorthand notation $\mathpzc{D}(\bm{M}) = \mathpzc{A}(\bm{M}) \otimes \mathpzc{B}_A(\bm{M}) \otimes \mathpzc{B}_B(\bm{M}) \otimes \mathpzc{B}_C(\bm{M})$ for the total decoupled *-algebra and call the coupled *-algebra $\mathpzc{C}(\bm{M})$. 

Because of the spacetime relationships of the regions $O_i$, we can place two non-intersecting spacelike Cauchy surfaces $\Sigma_1, \Sigma_2$, such that 
\begin{enumerate}
	\item $O_A \subset J^-(\Sigma_1)$ and $O_A \subset J^-(\Sigma_2)$,
	\item $O_B \subset J^+(\Sigma_1)$ and $O_B \subset J^-(\Sigma_2)$,
	\item $O_C \subset J^+(\Sigma_1)$ and $O_C \subset J^+(\Sigma_2)$.
\end{enumerate}
Dividing $M$ into regions $M_A = J^-(\Sigma_1), M_B = J^+(\Sigma_1) \cap J^-(\Sigma_2), M_C = J^+(\Sigma_2)$, we obtain *-algebras $\mathpzc{D}(\bm{M}; M_i)$. By the time-slice property of AQFT
\begin{equation} \label{timeslice_Sorkin_scn} \mathpzc{D}(\bm{M}; M_i) = \mathpzc{D}(\bm{M}) \; \; \; \forall i \in \{A, B, C\}. \end{equation}
Therefore, we can use the procedure from section \ref{operations} in each region $M_i$ to get scattering morphisms $\tilde{\Theta}_i$, which are automorphisms on $\mathpzc{A}(\bm{M}_i) \otimes \mathpzc{B}_i(\bm{M}_i)$ that relate the decoupled theory identified with the coupled one in the region $M^+(K_i) \cap M_i$ to the decoupled theory identified with the coupled one in the region $M^-(K_i) \cap M_i$. Now define $\Theta_i$ to be the automorphisms on $\mathpzc{D}(\bm{M}_i)$, which act trivially on $\mathpzc{B}_{j \neq i}(\bm{M}_i)$ and as $\tilde{\Theta}_i$ on $\mathpzc{A}(\bm{M}_i) \otimes \mathpzc{B}_i(\bm{M}_i)$. By eq. \ref{timeslice_Sorkin_scn}, these are also automorphisms on $\mathpzc{D}(\bm{M})$. This is illustrated in Figure \ref{fig_sorkin_2}.

If Charlie measures an induced system observable $C$ using his probe, given an initial state of the system and Alice's and Bob's probes $\omega \otimes \sigma_A \otimes \sigma_B$, the expectation value will be
\begin{equation} 
\langle C \rangle = (\omega \otimes \sigma_A \otimes \sigma_B)\big((\Theta_A \circ \Theta_B)(C \otimes \mathds{1} \otimes \mathds{1})\big).
\end{equation}
This relation is not immediately obvious and it has been rigorously proven in \cite{bostelmann}. The reasoning relies on the idea to combine the probe theories into a single probe theory with *-algebra $\mathpzc{B}(\bm{M}) = \mathpzc{B}_A(\bm{M}) \otimes \mathpzc{B}_B(\bm{M}) \otimes \mathpzc{B}_C(\bm{M})$, interaction region $K = K_A \cup K_B \cup K_C$, scattering morphism $\Theta$ and initial probe state $\sigma = \sigma_A \otimes \sigma_B \otimes \sigma_C$. We are measuring $\tilde{B}_C = \mathds{1} \otimes \mathds{1} \otimes B_C$ at late times. Here $B_C$ is the observable on Charlie's probe that induces the observable $C$ on the system, given the initial probe state $\sigma_C$. From the discussion above, we can write $\Theta = \Theta_A \circ \Theta_B \circ \Theta_C$, because of the causal ordering imposed on the observers by $\Sigma_1, \Sigma_2$. 

Now
\begin{equation} \begin{split}
\langle C \rangle &= \big(\omega \otimes \sigma)(\Theta(\mathds{1} \otimes \tilde{B}_C)\big) \\
&= (\omega \otimes \sigma_A \otimes \sigma_B \otimes \sigma_C)\\
&\;\;\;\;\;\;\;\;\;\;\;\;\;\;\big((\Theta_A \circ \Theta_B \circ \Theta_C)(\mathds{1} \otimes \mathds{1} \otimes \mathds{1} \otimes B_C)\big) \\
&= (\omega \otimes \sigma_A \otimes \sigma_B)\big((\Theta_A \circ \Theta_B) (C \otimes \mathds{1} \otimes \mathds{1})),
\end{split} \end{equation}
where trivial action of $\Theta_i$ on the probe theories with index $j \neq i$ is used, together with the fact the $C$ is the system observable induced by $B_C$ using Charlie's probe coupling and initial probe state $\sigma_C$. 

Causality says, that results of measurements obtained by Charlie should be independent of the operations preformed by Alice. Hence,
\begin{multline} \label{causality_FV} (\Theta_A \circ \Theta_B)(C \otimes \mathds{1} \otimes \mathds{1}) =  \Theta_B (C \otimes \mathds{1} \otimes \mathds{1}) \\
\forall C \in \mathpzc{A}(\bm{M}; O_C). \end{multline}
The main result of \cite{bostelmann} is that this is always the case. Here I give a sketch of the proof. 

Consider the region $T \coloneqq J^-(\bar{O}_C) \cap \Sigma_1$, where $\bar{O}_C$ is the compact closure of $O_C$. Since $O_C$ and $O_A$ are spacelike separated, we expect that $T \cap (J^+(\bar{O}_A) \cap \Sigma_1) = \O$, and hence that $T \subset K_A^\perp$. Furthermore, by definition, $\bar{O}_C \subseteq D(T) \subseteq D(M^-(O_B) \cap O_A^\perp)$. Since $C \otimes \mathds{1} \otimes \mathds{1}$ can be localized in $O_C$, using property 3 of the scattering morphism in Lemma \ref{scattering_properties}, $\Theta_B(C \otimes \mathds{1} \otimes \mathds{1})$ can be localized in $M^-(O_B) \cap O_A^\perp$. Now by property 2 in Lemma \ref{scattering_properties}, $\Theta_A$ has to act trivially on $\Theta_B(C \otimes \mathds{1} \otimes \mathds{1})$, which proves that eq. \ref{causality_FV} is always satisfied in the Sorkin scenario. The geometric arguments in this proof can be formalized, which is the content of Lemmas 3 and 4 in \cite{bostelmann}.

\paragraph{Comments on causality in the FV framework.} From the discussion above, we see that if we formalize quantum operations as interactions between system and probe quantum fields in the FV language, causality is always respected in the Sorkin scenario. It is shown in \cite{bostelmann} that this result generalizes beyond three observers by the process of causal factorization to any collection of causally orderable observers. There are some other important points about causality however. 

Firstly, it is worth noting that if the probe observable can be localized in $K^\perp$, the induced observable is a multiple of the identity (Theorem 3.3 in \cite{fewster}). This can be checked explicitly. Suppose $B \in \mathpzc{B}(\bm{M}; L)$, where $L \in K^\perp$. Now
\begin{equation} \varepsilon_\sigma(B) = \eta_\sigma(\Theta(\mathds{1} \otimes B)) = \eta_\sigma(\mathds{1} \otimes B) = \sigma(B) \mathds{1}, \end{equation}
where I used property 2 of the scattering morphism from Lemma \ref{scattering_properties}. This means that measuring probe observables in the causal complement of the interaction region gives us no information about the system. 

Secondly, one might be concerned about the situation when the probe observables $B_1 \in \mathpzc{B}(\bm{M}; N_1)$ and $B_2 \in \mathpzc{B}(\bm{M}; N_2)$ are measured, such that $N_1$ and $N_2$ are causally disjoint. The induced observables can be both localized in $K$. Einstein causality \ref{einstein} in section \ref{aqft} demands that $[B_1, B_2] = 0$. Does this mean we are restricted to measuring only commuting observables on the system in this setup? Luckily, $\varepsilon_\sigma$ is not an isomorphism, due to the definition of $\eta_\sigma$, so the fact that $[B_1, B_2] = 0$ does \textit{not} imply that the induced observables have to commute too. 

Finally, I would like to discuss the situation when the observers cannot in fact be causally ordered, so that causal factorization cannot be applied directly. Suppose in particular that there are two observers, $\mathcal{O}_1$ and $\mathcal{O}_2$, with probes associated with *-algebras $\mathpzc{B}_1(\bm{M})$ and $\mathpzc{B}_2(\bm{M})$ respectively, coupled to the system AQFT with *-algebra $\mathpzc{A}(\bm{M})$ in compact regions $K_1$ and $K_2$ respectively. Furthermore, suppose that there exists no spacelike Cauchy surface $\Sigma$, such that $K_1 \subset J^-(\Sigma)$ and $K_2 \subset J^+(\Sigma)$ or vice versa, so the interactions cannot be causally ordered.\footnote{Note that $K_1$ and $K_2$ can overlap.} Suppose that $\mathcal{O}_1$ measures an EVM $E_1: \chi_1 \rightarrow \mathpzc{B}_1(\bm{M})$ ad that $\mathcal{O}_2$ measures an EVM $E_2: \chi_2 \rightarrow \mathpzc{B}_2(\bm{M})$. 

To resolve this problem we can, as in the previous section, combine the probes to get a single probe associated with a *-algebra $\mathpzc{B} = \mathpzc{B}_1 \otimes \mathpzc{B}_2$, coupled with the system in a compact region $K$, such that $K_i \subseteq K$ for $i \in \{1, 2\}$. We can now consider measuring the joint EVM $E: \chi_1 \times \chi_2 \rightarrow \mathpzc{B}_1 \otimes \mathpzc{B}_2$, such that $E(X \times Y) = E_1(X) \otimes E_2(Y) \; \forall X \in \chi_1, Y \in \chi_2$. We managed to combine the observers to get a single observer, which can now be causally factorized in a larger collection of observers.

In this way, given a collection of observers who cannot be causally ordered, we can always combine some of these observers to get a smaller collection that can be causally ordered and we use causal factorization to show that causality is obeyed. Note that if two observers cannot be causally ordered, it means that there are timelike curves between their corresponding regions in either direction, so causality doesn't constrain operations that these observers can perform. Using this argument, it is clear that causal factorization is enough to ensure proper causal behaviour in the FV framework. 

\end{multicols}

\begin{figure}
\centering
\includegraphics{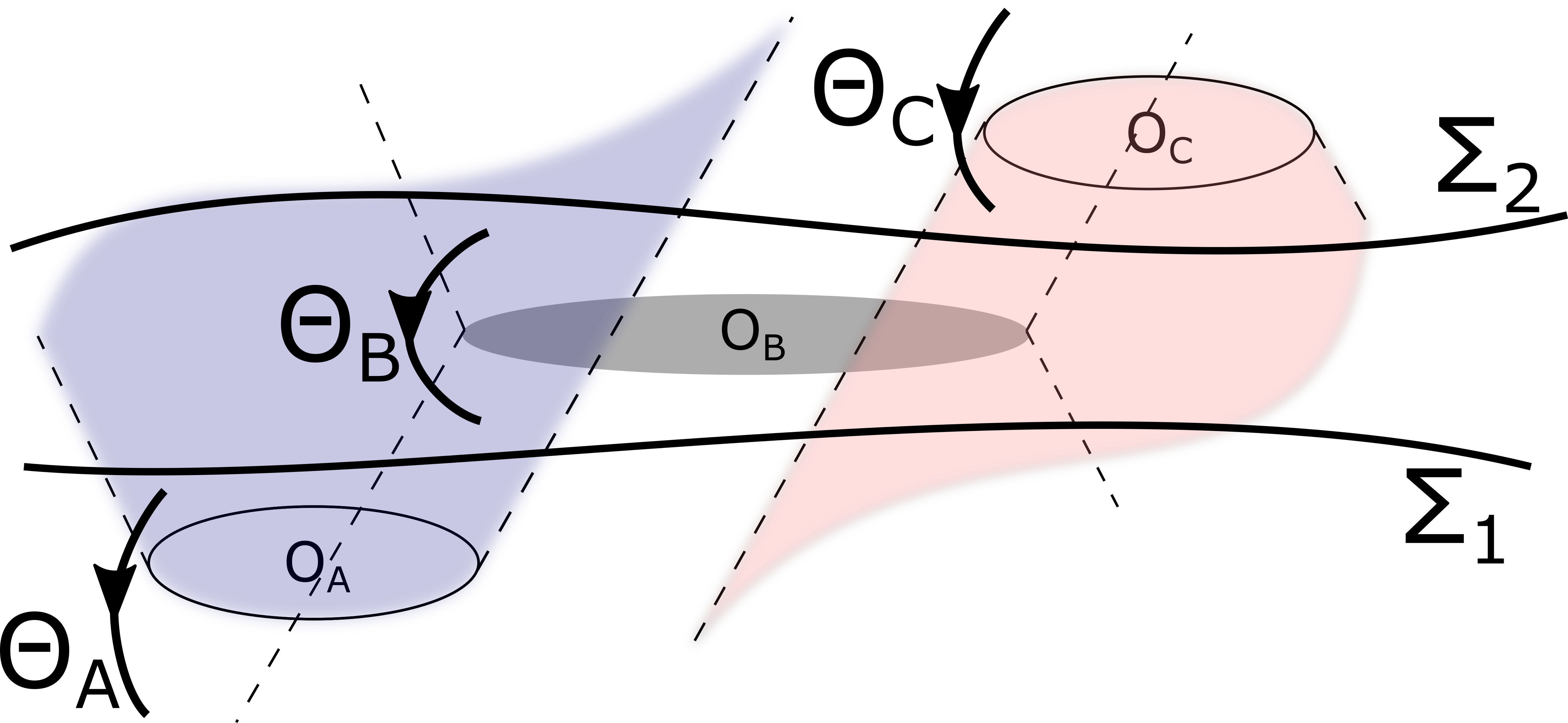}
\caption{\textbf{The Sorkin scenario in the FV framework.} $O_A, O_B$ and $O_C$ are regions containing the interaction regions of the probes controlled by Alice, Bob and Charlie respectively. $\Sigma_i$ are Cauchy surfaces. This figure illustrates how the scattering morphisms $\Theta_i$ act on the uncoupled theory. The dashed lines are null curves. The blue region is $J^+(O_A)$, the red region is $J^-(O_C)$. Time is vertically upwards, one spatial dimension is on the horizontal axis.}
\label{fig_sorkin_2}
\end{figure}
\hrulefill
\begin{multicols}{2}

\section{Sum over histories approach} \label{sum_over_hist}
In this section I will study causality in the language of Feynman path integrals and sum over histories. First I will introduce this idea in the case of a free particle and then move to QFT. Path integral approach has been suggested in \cite{sorkin} and \cite{borsten} as the fundamental picture that would iron out all difficulties with causality standard exposition of QFT contains. However, I am not aware of any explicit attempts to do so and I will show in this section why I don't think this approach provides much new insight into the problem. 

In sections \ref{free_particle} to \ref{decoherent_hist} I will develop ideas and notations from the literature. In section \ref{causality_in_decoherent_hist} I give my own definitions and results, applying the preceding material to AQFT and the study of causality.

\subsection{Free particle} \label{free_particle}
\textit{[In this section I am working in natural units $\hbar = c = 1$.]}

\vspace{2.5mm}

\noindent The path integral formulation of quantum mechanics puts paths through configuration space in the spotlight. See e.g. chapter 1 of \cite{brown_qft} for a thorough introduction. For a free particle in 1 spatial dimension $q$, given at $t=0$ an initial wavefunction $\psi(q, 0)$, the path integral yields a wavefunction at time $T>0$ as
\begin{equation} \label{pathint} 
\psi(q, T) = \int dq' \psi(q', 0) \int_{[q, q']} \mathcal{D}q \: e^{iS[q]},
\end{equation}
where $S[q(t)]$ is the classical action. The path integral $\int_{[q, q']} \mathcal{D}q$ denotes an integral over all paths $q(t)$ through the 1-dim. configuration space of the particle parametrized by time, such that $q(0) = q'$ and $q(T) = q$.

This time evolution works in the free case, where there is no measurement in the time interval $(0, T)$. However, if the particle is measured to be in some regions $\Delta^1_{\alpha_1}, \Delta^2_{\alpha_2}, ..., \Delta^n_{\alpha_n}$ at times $t_1, t_2, ..., t_n$, such that $0 < t_i < T \; \forall i = 1, ..., n$, the projection postulate translates to the path integral approach by restricting the integration to paths in the set $C_\alpha$, which satisfy the boundary conditions as before, but furthermore $q(t_i) \in \Delta^i_{\alpha_i} \; \forall i = 1, ..., n$. Here, $\alpha = (\alpha_1, ..., \alpha_n)$. Write this as
\begin{equation} \psi(q, T) = \braket{q | \psi(T)} = \int dq' \psi(q', 0) \int_{C_\alpha} \mathcal{D}q \: e^{iS[q]}. \end{equation}
For a thorough analysis of this principle, see \cite{caves}. 

It is also possible to generalize this to measurements of observables other than position, as explained in \cite{hartle}. Suppose a measurement of a functional $F[q(t)]$. If the measured value is within a real interval $\Delta_\alpha$, the set of paths over which we integrate becomes
\begin{equation} C_\alpha = \{q(t): F[q(t)] \in \Delta_\alpha, q(0) = q', q(T) = q\}. \end{equation}

The position $q$ is still the dynamical variable rather than a spacetime coordinate (recall section \ref{dyn_vars_vs_coords}), so the theory remains non-relativistic. It is parametrized by a single global time variable and the time order is the same for all observers.

\subsection{QFT} \label{path_int_qft}
We can generalize the above approach by identifying the configuration space of a scalar quantum field $\phi: M \rightarrow \mathbb{R}$. The process is described in detail in chapters 2 and 3 of \cite{brown_qft}. Some care should be taken when generalizing to spinor fields or gauge fields, but in this discussion this is an unnecessary complication. See e.g. \cite{banks_qft} for a thorough exposition. It will be beneficial to chose coordinates, such that the first coordinate (call it time $t$) generates timelike curves if we fix the other three and inherits the time orientation from $\bm{M}$. This is always possible in a globally hyperbolic Lorentzian manifold. Now we can write $\phi = \phi(\bm{x}, t)$, where $\bm{x}$ are the three spacelike coordinates. I will call the space of all functions $\mathbb{R}^3 \rightarrow \mathbb{R}$ the \textit{configuration space} of the field. The time coordinate generates paths through this space, similarly to a free particle in 1 dimension. The difference is that now I have at each time infinitely many ``positions", each corresponding to the value of $\phi$ at some point in space.\footnote{I will sometimes use the word \textit{space} in this section to address the spacelike hypersurface in $M$ generated by fixing the time coordinate. It should be clear, where something else is meant by it.} 

It is problematic to integrate over values of the field at \textit{all} space points. That is why sources are introduced and only vacuum-to-vacuum transitions considered.

Consider a field with classical action $S$ given by Lagrangian density $\mathcal{L}$ through
\begin{equation} S_0[\phi] = \int d^4x \: \mathcal{L}(\phi, \partial \phi, ...). \end{equation}
The source $J(x)$ is introduced, such that the source dependent action is given by
\begin{equation} S_J[\phi] = \int d^4x \: \big(\mathcal{L}(\phi, \partial \phi, ...) + J(x) \phi(x)\big). \end{equation}

Let us pick a vacuum configuration $\phi_0(\bm{x})$, corresponding to a state $\ket{0}$ in the Fock space of the field and write $\ket{0, t}$ for the time evolution of a field state with $\ket{0, - \infty} = \ket{0}$. I would like to now get a path integral expression for the vacuum-to-vacuum transition probability. Studying eq. \ref{pathint} and generalizing it to the field configuration space, given vacuum initial state $\ket{0, - \infty} = \ket{0}$, I get
\begin{equation} \label{path_integral_qft}
\braket{0 | 0, \infty}_J = \frac{1}{\mathcal{Z}_0} \int_{[\phi_0, \phi_0]} \mathcal{D}\phi \: e^{iS_J[\phi]} \eqqcolon \mathcal{Z}[J], 
\end{equation}
where the normalization
\begin{equation} \mathcal{Z}_0 = \int_{[\phi_0, \phi_0]} \mathcal{D}\phi \: e^{iS_0[\phi]} \end{equation}
is inserted since $\phi_0$ is assumed to have zero energy, and so no time evolution and hence I require $\braket{0 | 0, \infty}_{J=0} = 1$. The path integral is over all paths through configuration space $\phi(\bm{x}, t)$, such that $\phi(\bm{x}, \pm\infty) = \phi_0(\bm{x})$. Taking functional derivatives of $\mathcal{Z}[J]$ with respect to the sources and setting $J = 0$ gives time ordered expectation values of field operators, i.e. we can use them to set different initial and final conditions. For example
\begin{equation}
\frac{1}{i^n} \frac{\delta^n \mathcal{Z}[J]}{\delta J(x_1) ... \delta J(x_n)} = \braket{0 | \mathcal{T}[\phi(x_1) ... \phi(x_n)]}.
\end{equation}
For further details of this procedure, see e.g. \cite{brown_qft}. 

The projection postulate from the free particle picture generalizes to quantum fields. If an observable $F[\phi(\bm{x}, t)]$ is measured to be in a real interval $\Delta_\alpha$, the path integral will run over only a subset of paths through the configuration space $C_\alpha$, where
\begin{equation} 
C_\alpha = \{\phi(\bm{x}, t): F[\phi(\bm{x}, t)] \in \Delta_\alpha, \phi(\bm{x}, \pm\infty) = \phi_0(\bm{x})\}. 
\end{equation}
So we now write the vacuum-to-vacuum transition probability
\begin{equation} \mathcal{Z}[J] = \frac{1}{\mathcal{Z}_0} \int_{C_\alpha} \mathcal{D}\phi \: e^{iS_J[\phi]}. \end{equation}

It is worth noting that I have chosen a preferred set of coordinates for the description. This choice enters the formalism only through the boundary conditions. Lagrangian density is a scalar that is integrated over all spacetime and the path integral can be understood as an integral over configurations of the scalar $\phi(\mathrm{x})$. However, since the boundary conditions are given at temporal infinities, they won't have any effect on the general story. There is also the issue of the Unruh effect, which says that the vacuum is observer dependent, even in flat Minkowski spacetime, see \cite{crispino}. I will however skip these subtleties, as they do not change the discussion; and so I will assume that all observers can agree on a preferred vacuum state and on the exact form of the boundary conditions, e.g. the inertial observer vacuum at their temporal infinities, in the case of the Unruh effect.

\subsection{Decoherent histories} \label{decoherent_hist}
In the above picture, each path through configuration space is assigned a \textit{complex} amplitude, which are then summed up to get the total probability of transition between the initial and final state. Therefore there is interference between the paths and we cannot assign them classical real additive probabilities. The concept of \textit{decoherent histories} is a way of recovering the additive classical probabilities for larger bundles of paths, which don't interfere. For a thorough introduction see \cite{omnes} and \cite{hartle}. 

Consider again the case of a free particle for simplicity. Let $\ket{\psi}$ be an initial state. I would like to describe a history of the particle as: ``The value of an observable $A_1$ at time $t_1$ is in the interval $\Delta^1_{\alpha_1}$, the value of an observable $A_2$ at time $t_2$ is in the interval $\Delta^2_{\alpha_2}$ etc." Each sequence $\alpha = (\alpha_1, ..., \alpha_n)$ now defines a distinct history. Let $P^i_{\alpha_i}$ for $i = 1, ..., n$ be the Schr\"odinger picture projector of observable $A_i$ corresponding to its value in the interval $\Delta_{\alpha_i}$. Suppose that the alternatives are exhaustive and mutually exclusive, so that
\begin{equation} \label{projectors}
\sum_{\alpha_i} P^i_{\alpha_i} = 1 \;\;\; \text{ and } \;\;\; P^i_{\alpha_i} P^i_{\beta_i} = \delta_{\alpha_i \beta_i}P^i_{\alpha_i}. 
\end{equation}
We now say that, in each history $\alpha$, the state is time evolved in the Schr\"odinger picture by the operator expressing the definite values at the various times $t_i$
\begin{multline}
C_{\alpha} \coloneqq e^{-iH(T-t_n)}P^n_{\alpha_n}e^{-iH(t_n - t_{n-1})} ...\\... e^{-iH(t_2-t_1)}P^1_{\alpha_1}e^{-iHt_1},
\end{multline}
where $H$ is the Hamiltonian. This can be rewritten using the Heisenberg picture projectors as
\begin{equation} 
C_{\alpha} = e^{-iHT}P^n_{\alpha_n}(t_n)...P^1_{\alpha_1}(t_1),
\end{equation}
where $P^i_{\alpha_i}(t) \coloneqq e^{iHt}P^i_{\alpha_i}e^{-iHt}$ are the time evolved Heisenberg picture projectors. The state $\ket{\psi_\alpha} \coloneqq C_{\alpha}\ket{\psi}$ is not normalized and its square norm is the probability of the history
\begin{equation} P(\alpha) = \braket{\psi_\alpha | \psi_\alpha}. \end{equation}

Each vector $\ket{\psi_\alpha}$ corresponds to a history. The histories are said to be \textit{decoherent} if
\begin{equation} \label{decoherence} \braket{\psi_{\alpha'} | \psi_\alpha} \approx 0 \;\;\; \forall \alpha' \neq \alpha. \end{equation}
This means that each history is disjoint from all the others and it makes sense to talk about them as being different alternatives for the dynamics of the system. 

The important consequence of definition eq. \ref{decoherence} is that it ensures additivity of probabilities. To see this, consider the $n=2$ case. Now
\begin{equation}\begin{split}\label{additivity}
\sum_{\alpha_1}P(\alpha_1, \alpha_2) &= \sum_{\alpha_1}\braket{\psi | P^1_{\alpha_1}(t_1) P^2_{\alpha_2}(t_2) P^2_{\alpha_2}(t_2) P^1_{\alpha_1}(t_1) | \psi} \\
&\approx \sum_{\alpha_1, {\alpha'}_1}\braket{\psi | P^1_{{\alpha'}_1}(t_1) P^2_{\alpha_2}(t_2) P^2_{\alpha_2}(t_2) P^1_{\alpha_1}(t_1) | \psi} \\
&= \sum_{\alpha_1}\braket{\psi | P^2_{\alpha_2}(t_2) P^2_{\alpha_2}(t_2) P^1_{\alpha_1}(t_1) | \psi} \\
&= \braket{\psi | P^2_{\alpha_2}(t_2) P^2_{\alpha_2}(t_2) | \psi} \\
&= P(\alpha_2),
\end{split}\end{equation}
where the second equality is by eq. \ref{decoherence} and the third and fourth equality by eq. \ref{projectors}.

It is important to emphasise that the histories do not correspond to a sequence of \textit{measurements}. These would in fact always enforce decoherence. For example, consider the standard double slit experiment. In the path integral formulation, paths passing through each slit interfere, giving rise to an interference pattern on the screen, giving the probabilities for where the particle lands. This means that the histories like: ``Particle passed through the first slit and then landed in the interval $\Delta$ on the screen," do not decohere, i.e. don't obey the sum rule. However, if we measure which slit the particle passed through, we enforce decoherence between these histories, the sum rule between the histories as before is obeyed, as indicated by the disappearance of the interference pattern.

\subsection{Causality in decoherent histories approach in AQFT} \label{causality_in_decoherent_hist}
In an AQFT $\mathpzc{A}(\bm{M})$, I will use the concept of \textit{projective effect valued measure} (PEVM) to generalize projector valued measures $\{P^i_{\alpha_i}\}$ from the previous section to the *-algebraic setting, as in section \ref{FV-measurement_scheme}. Recall the definition of PEVMs def. \ref{evm}. 

Suppose that two observers $\mathcal{O}_1$ and $\mathcal{O}_2$ perform measurements associated with PEVMs $E_1: \chi_1 \rightarrow \mathpzc{A}(\bm{M}; O_1)$ and $E_2: \chi_2 \rightarrow \mathpzc{A}(\bm{M}; O_2)$ respectively, where $\chi_1$ and $\chi_2$ are associated with sets of all possible measurement values $\Omega_1$ and $\Omega_2$ respectively. The elements of $\chi_i$, as in section \ref{FV-measurement_scheme}, correspond to the intervals in which the values of observables measured by $\mathcal{O}_i$ can be found. Suppose also that there exists a spacelike Cauchy surface, such that $O_1$ is in its causal past and $O_2$ is in its causal future. Now we can choose coordinates, such that there is a coordinate $t$, which is constant on the Cauchy surface, parametrizes timelike curves and in which $O_1$ is before $O_2$, so that we can apply the coordinate dependent formalism from section \ref{path_int_qft}. This generates a set of histories associated with operators
\begin{equation} C_\alpha = E_2(\alpha_2) E_1(\alpha_1), \end{equation}
labelled by $\alpha = (\alpha_1, \alpha_2)$. These operators can be localized in any compact region containing $O_1$ and $O_2$, due to the Haag property (condition \ref{haag} in section \ref{aqft}). We can now generalize the notion of decoherent histories to
\begin{equation} C_{\alpha'}^* C_\alpha = 0, \end{equation}
for all $\alpha, \alpha'$ such that $\alpha'_i \cap \alpha_i = \O$ for some $i$. Given a state $\omega$, the probability of a history associated with the sequence $\alpha$ is
\begin{equation} P(\alpha) = \omega(C_\alpha^*C_\alpha). \end{equation}

If we do not allow any other communication between the observers, the observer $\mathcal{O}_2$ cannot know the result of the measurement $\mathcal{O}_1$ performs. Hence if $\mathcal{O}_1$ does perform a measurement associated with the subset $A_1 \subset \chi_1$, such that $\alpha_1 \cap \beta_1 = \O$ for all $\alpha_1, \beta_1 \in A_1$ such that $\alpha_1 \neq \beta_1$ and $\bigcup_{\alpha_1 \in A_1} \alpha_1 = \Omega_1$, the probability distribution of results $\mathcal{O}_2$ observes will be 
\begin{equation} P(\alpha_2 | A_1) = \sum_{\alpha_1 \in A_1} P(\alpha_1, \alpha_2) = \sum_{\alpha_1 \in A_1} \omega(C_{(\alpha_1, \alpha_2)}^*C_{(\alpha_1, \alpha_2)}). \end{equation}
On the other hand if $\mathcal{O}_1$ doesn't perform their measurement, the probability distribution $\mathcal{O}_2$ observes will be
\begin{equation} P(\alpha_2) = \omega\big((E_2(\alpha_2)^*E_2(\alpha_2)\big) = \omega(E_2(\alpha_2)), \end{equation}
where the last equality comes from the definition of PEVMs def. \ref{evm}.

Hence $\mathcal{O}_1$ will be able to send a message (by either performing the measurement or not) to $\mathcal{O}_2$ if
\begin{equation} \label{signalling} P(\alpha_2 | A_1) \neq P(\alpha_2), \end{equation}
which can be written explicitly as
\begin{equation} \omega(E_2(\alpha_2)) \neq \sum_{\alpha_1 \in A_1} \omega(C_{(\alpha_1, \alpha_2)}^*C_{(\alpha_1, \alpha_2)}). \end{equation}
Since we want a condition on the \textit{measurements themselves}, we don't want it to depend on the state $\omega$. Hence we need the condition eq. \ref{signalling} to hold for all $\omega$ and the condition on the measurements to be able to transfer information from $\mathcal{O}_1$ to $\mathcal{O}_2$ becomes
\begin{equation} \begin{split}
E_2(\alpha_2) &\neq \sum_{\alpha_1 \in A_1} C_{(\alpha_1, \alpha_2)}^*C_{(\alpha_1, \alpha_2)} \\
\text{or }&\text{equivalently}\\
E_2(\alpha_2) &\neq \sum_{\alpha_1 \in A_1} E_1(\alpha_1) E_2(\alpha_2) E_1(\alpha_1) .
\end{split} \end{equation}

By comparing the condition in eq. \ref{signalling} with eq. \ref{additivity} we notice that the measurements won't be able to signal for any choice of $A_1$ if the histories they generate decohere. In fact, studying the derivation of eq. \ref{additivity}, it will be enough to demand that the histories \textit{decohere in the first observable} only. I mean by this that 
\begin{equation} C^*_{(\alpha'_1, \alpha_2)} C_{(\alpha_1, \alpha_2)} = 0 \;\;\; \forall \alpha_2, \alpha'_1 \cap \alpha_1 = \O. \end{equation}
Hence we define decoherent histories in AQFT as

\begin{definition}[Decoherent histories in AQFT]\label{def_decoherent_hist_aqft}
Suppose an AQFT $\mathpzc{A}(\bm{M})$ and a set of $n$ PEVMs $\{E_i\}_{i=1}^n$, such that $E_i: \chi_i \rightarrow \mathpzc{A}(\bm{M}; O_i)$. Suppose further that there exists a set of $n-1$ non-intersecting spacelike Cauchy surfaces $\{\Sigma_i\}_{i=1}^{n-1} \in M$, such that for each $1 \leq i \leq n-1$ all $O_j$ with $j \leq i$ are in the causal past of $\Sigma_i$ and all $O_k$ with $k > i$ are in the causal future of $\Sigma_i$. Now let us choose a coordinate system in $\bm{M}$, such that there is a coordinate $t$, which is constant on each $\Sigma_i$, increases in the direction of the time orientation on $\bm{M}$ and the regions $O_i$ are ordered in $t$. 

The histories $C_\alpha \coloneqq E_n(\alpha_n) ... E_1(\alpha_1)$ for each $\alpha = (\alpha_1, ..., \alpha_n)$ are said to be \textit{decoherent in the i-th observable} if
\begin{equation} \label{aqft_partial_dch} C_{(\alpha_1, ..., \alpha^\prime_i, ..., \alpha_n)}^* C_{(\alpha_1, ..., \alpha_i, ..., \alpha_n)} = 0 \;\; \forall \alpha_{j \neq i}, \alpha'_i \cap \alpha_i = \O. \end{equation}

Furthermore, the histories are called \textit{decoherent} if
\begin{equation} C^*_{\alpha'}C_\alpha = 0 \end{equation}
for all $\alpha', \alpha$ such that $\alpha'_i \cap \alpha_i = \O$ for some $i$.
\end{definition}

Let us consider the case of two observers again. Suppose $O_1$ and $O_2$ are in fact causally disjoint. Now causality demands that they can't be able to signal each other. This is ensured by Einstein causality in AQFT (condition \ref{einstein} in section \ref{aqft}), which implies that
\begin{equation} [E_1(\alpha_1), E_2(\alpha_2)] = 0. \end{equation}
This is enough to ensure decoherence in the first observable, since
\begin{equation} \begin{split} \label{two_obs_causality}
(C_{(\alpha'_1, \alpha_2)})^* C_{(\alpha_1, \alpha_2)} &= E_1(\alpha'_1) E_2(\alpha_2) E_2(\alpha_2) E_1(\alpha_1) \\
&= E_1(\alpha'_1) E_1(\alpha_1) E_2(\alpha_2) \\
&= E_1(\alpha'_1 \cap \alpha_1) E_2(\alpha_2)
\end{split} \end{equation}
vanishes if $\alpha'_1 \cap \alpha_1 = \O$.

The Sorkin scenario, however, introduces a complication. Decoherence becomes non-trivial, since Bob's projectors don't commute with either Alice's, nor Charlie's. If we give Alice, Bob and Charlie the PEVMs $E_i: \chi_i \subset \mathpzc{A}(\bm{M}; O_i)$ with $\chi_i$ associated with sets $\Omega_i$ and $i = A, B, C$ respectively, the condition on the PEVMs so that Alice cannot signal Charlie becomes
\begin{equation} \sum_{\alpha_A \in A_A} P(\alpha_A, \alpha_B, \alpha_C) = P(\alpha_B, \alpha_C), \end{equation}
for all $A_A$ such that $\alpha_A \cap \beta_A = \O$ for all $\alpha_A, \beta_A \in A_A$ such that $\alpha_A \neq \beta_A$ and $\bigcup_{\alpha_A \in A_A} \alpha_A = \Omega_A$, or explicitly
\begin{equation} \label{causality_sorkin} \begin{split}
E_A(\alpha'_A) E_B&(\alpha_B) E_C(\alpha_C) E_C(\alpha_C) E_B(\alpha_B) E_A(\alpha'_A) = 0 \\
&\forall \alpha_B, \alpha_C, \alpha'_A \cap \alpha_A = \O, 
\end{split} \end{equation}
since in the Sorkin scenario we can order Alice, Bob and Charlie uniquely in the sense of def. \ref{def_decoherent_hist_aqft}. 

Even more worryingly, it seems we should be able to introduce another observer, Beatrice, localized in a region $O_{B'}$, such that $O_{B'} \cap J^+(O_A) \neq \O$, $O_{B'} \cap J^-(O_C) \neq \O$, $O_{B'} \cap O_B = \O$ and that there is a spacelike Cauchy surface with $O_B$ in its causal past and $O_{B'}$ in its causal future, with PEVM $E_{B'}: \chi_{B'} \rightarrow \mathpzc{A}(\bm{M})$, where $\chi_{B'}$ is associated with $\Omega_{B'}$. This will result in an even stricter condition than eq. \ref{causality_sorkin}. The new condition will be
\begin{equation} \begin{split}
C^*_{(\alpha'_A, \alpha_B, \alpha_{B'}, \alpha_C)} C_{(\alpha_A, \alpha_B, \alpha_{B'}, \alpha_C)} = 0 \\
\forall \alpha_B, \alpha_{B'}, \alpha_C, \alpha'_A \cap \alpha_A = \O. 
\end{split} \end{equation}

We can keep adding observers measuring observables localized in the region bounded (in ``time") by some non-intersecting spacelike Cauchy surfaces, for which Alice is in their causal past and Charlie in their causal future, tightening the conditions on Alice not to be able to signal Charlie even further. The only other condition I impose on the newly introduced observables is that I have to be able to order them by non-intersecting spacelike Cauchy surfaces in the sense of def. \ref{def_decoherent_hist_aqft}. I will call this setup the \textit{extended Sorkin scenario} and the observables other than $E_A$ and $E_C$ the \textit{intermediate observables}. 

Unfortunately, if histories in the extended Sorkin scenario with $n-1$ intermediate observers decohere in Alice's observable, it is \textit{not} guaranteed that they will decohere if an $n$-th intermediate observer is included. The PEVMs don't commute in general; and as is the case even for (non-orthogonal) projectors $P, P_1, P_2$ on $\mathbb{R}^n$, given $P_1 P_2 = 0$, it may well be the case that $P_1 P P_2 \neq 0$.

Therefore even in the language of decoherent histories, the conditions imposed by causality on observables in AQFT are highly non-trivial. This language however provides us with a simple check on ``who can signal who" in a particular spacetime setting of observers measuring localized observables, provided they can be ordered in the sense of def. \ref{def_decoherent_hist_aqft}. This is summarized in the following claim.

\begin{claim} \label{signalling_histories}
Given a set of observers $\mathcal{O}_i$, each measuring a PEVM $E_i: \chi_i \rightarrow \mathpzc{A}(\bm{M}; O_i)$, which are ordered in the sense of def. \ref{def_decoherent_hist_aqft}, observer $\mathcal{O}_j$ will not be able to signal the observer $\mathcal{O}_k$ with $j < k$ if the histories $C_\alpha = E_k(\alpha_k) ... E_j(\alpha_j)$ for $\alpha = (\alpha_j, ..., \alpha_k)$ decohere in the observable corresponding to $E_j$. 
\end{claim}

This claim is justified by noting that observers ordered before $\mathcal{O}_j$ and after $\mathcal{O}_k$, don't provide further constraints on causality between $\mathcal{O}_j$ and $\mathcal{O}_k$. 

To see that, consider an observer $\mathcal{O}_m$ ordered after $\mathcal{O}_k$. The probability distribution of measurement of $E_k$ is determined just by histories ending at the measurement of $E_k$ and we can completely ignore $E_m$, by definition of what we mean by a probability distribution for measurement of $E_k$. 

If instead we consider $\mathcal{O}_n$ ordered before $\mathcal{O}_j$, histories $C_\alpha = E_k(\alpha_k) ... E_j(\alpha_j) E_n(\alpha_n)$ will decohere in the observable corresponding to $E_j$, given that histories $C_\alpha = E_k(\alpha_k) ... E_j(\alpha_j)$ decohere in the observable corresponding to $E_j$. This is an immediate consequence of def. \ref{def_decoherent_hist_aqft}. 

Note that the claim \ref{signalling_histories} implies that an observer will \textit{not} be able to signal the immediately following observer if their corresponding observables commute (following logic similar to eq. \ref{two_obs_causality}). This is consistent with the usual notion of non-signalling in quantum mechanics.

The problem with getting a more refined condition on causality in this formalism is apparent from the necessity of ordering the observables. It may well happen that ordering of some localized observables in the sense of def. \ref{def_decoherent_hist_aqft} is impossible. Now we cannot generate histories and all of the machinery above falls apart. The problem is, that the product of non-commuting self-adjoint operators is not necessarily self-adjoint, since for $E_1^* = E_1$ and $E_2^* = E_2$
\begin{equation} (E_1 E_2)^* = E_2^* E_1^* = E_2 E_1 \neq E_1 E_2. \end{equation}
Therefore products of PEVMs are not in general PEVMs themselves. It is not clear how to proceed in this language. However the discussion at the end of section \ref{FV-framework} shows that if we introduce the apparatus, there is in fact a joint PEVM localized in a region containing the regions where $E_1$ and $E_2$ are localized, which has $E_1$ and $E_2$ as its marginals. Hence the FV framework gives us a way of combining measurements that cannot be causally ordered. Hence claim \ref{signalling_histories} in fact does cover this eventuality: we need measurement theory to transform a given collection of observables into one which can be causally ordered and to which claim \ref{signalling_histories} can be applied.

\paragraph{Back to path integrals.} Now let us try to formulate causality conditions in the path integral language. Consider measurements of two functionals, $F_1[\phi]$ and $F_2[\phi]$, which depend on $\phi(\mathrm{x}), \partial \phi(\mathrm{x}), ...$ only for $\mathrm{x} \in O_1$ and $\mathrm{x} \in O_2$ respectively, where $O_1, O_2 \subset M$. Observer $\mathcal{O}_1$ is measuring $F_1$, observer $\mathcal{O}_2$ is measuring $F_2$. Suppose that $O_1$ and $O_2$ are spacelike separated, so causality dictates that $\mathcal{O}_1$ and $\mathcal{O}_2$ should not be able to signal each other.

A necessary condition for $\mathcal{O}_1$ not to be able to signal $\mathcal{O}_2$ is
\begin{equation} \label{path_int_causality} 
\Big\| \int_{C_{\alpha_2}} \mathcal{D}\phi \: e^{iS_0[\phi]} \Big\|^2  = \sum_{\alpha_1 \in A_1} \Big\| \int_{C_{(\alpha_1, \alpha_2)}} \mathcal{D}\phi \: e^{iS_0[\phi]} \Big\|^2, 
\end{equation}
where
\begin{equation} C_{\alpha_2} = \{\phi(\mathrm{x}): F_2[\phi(\mathrm{x})] \in \Delta_{\alpha_2}\} \end{equation}
and
\begin{equation} C_{(\alpha_1, \alpha_2)} = \{\phi(\mathrm{x}): F_1[\phi(\mathrm{x})] \in \Delta_{\alpha_1} \text{ and } F_2[\phi(\mathrm{x})] \in \Delta_{\alpha_2}\}. \end{equation}
The set $A_1$ is a set of labels of non-intersecting intervals $\Delta_{\alpha_1}$, which cover all possible values of $F_1[\phi]$.

It is not at all obvious that eq. \ref{path_int_causality} will be satisfied for arbitrary $F_1, F_2$. It is definitely not something guaranteed in the formalism. The path integral eq. \ref{path_integral_qft} is summing over \textit{all} the configurations $\phi(\mathrm{x})$, so there is no dynamics, from which causality could emerge, as in classical field theory, where it comes from the group velocity, limiting the speed at which disturbances can travel. Causality therefore has to be imposed \textit{a posteriori} (eq. \ref{path_int_causality}), just as in e.g. AQFT (where this is done either by including the apparatus as in the FV approach, section \ref{FV-framework}, or by considering the conditions on the histories, section \ref{causality_in_decoherent_hist}). At this stage, the path integral approach seems to me a bit more clumsy than the others for finding the right restriction on the observables that would ensure causality, in view of the difficulties of working with path integrals. Furthermore, there is not even a simple criterion equivalent to Einstein causality (condition \ref{einstein} in section \ref{aqft}), not to mention a condition that would prevent signalling in the Sorkin scenario. However, I would not dare to anticipate future developments of the solutions to this problem.

\section{Conclusions} \label{conclusions}
Causality is \textit{not} intrinsically present in either usual quantum mechanics, or quantum field theory. In the case of ordinary quantum mechanics, we do not even have all the spacetime coordinates included in the formalism, which prevents us from even formulating what is meant by relativistic causality in this theory (section \ref{dyn_vars_vs_coords}). This applies even if we use the sum over histories approach (section \ref{sum_over_hist}). I have examined, what quantum operations can we perform, if we use the classical notion of locality and quantize only the internal degrees of freedom (section \ref{causality_from_classical_loc}). In this case, causality follows from imposing locality on the operations.

Quantum field theory is therefore the right language, in which to think about quantizing causality. However, no matter whether a QFT is formulated in a Fock space, in terms of path integrals (section \ref{causality_in_decoherent_hist}) or as an AQFT (section \ref{aqft}), causality has to be imposed \textit{a posteriori}. An important practical setup in which to consider causality and which amplifies the main problem is the Sorkin scenario (section \ref{sorkin_scn}). Fewster and Verch proposed in \cite{fewster} a way to formalize measurement in AQFT as an interaction between the system and a probe (summarized in section \ref{FV-framework}). Their construction, the FV framework, localizes measurement and leads to a causal theory (section \ref{FV-causality}). This is a very satisfying result and one could argue that this is all we need. 

However, it is still an important practical problem to provide conditions on observables in QFT, which guarantee that their measurement won't violate causality. The task of finding such conditions is highly non-trivial and I haven't found a satisfying answer in the literature. I used the notion of decoherent histories to find quite a simple condition that checks for causality in a particular setting, where the observers can be ordered in a particular sense using non-intersecting spacelike Cauchy surfaces (section \ref{causality_in_decoherent_hist}). The generalization to a condition on the observables is however not simple.

\noindent\hrulefill

\printbibliography

\end{multicols}

\end{document}